\documentclass[a4paper,USenglish]{lipics-v2016}
 
\usepackage{microtype}
\usepackage{tikz}
\usepackage{xspace}
\usepackage{booktabs}

\makeatletter

\newcommand{\r@rrow}[3]{%
  \newcommand{#1}[2][]{%
    \def\next{#2\@ifempty{##1}{}{_{##1}}\@ifempty{##2}{}{^{##2}}}%
    \mathchoice{#3[##1]{##2}}{\next}{\next}{\next}%
  }%
}
\newcommand{\l@rrow}[3]{%
  \newcommand{#1}[2][]{%
    \def\next####1{%
      \setbox0=\hbox{$####1\vphantom{#2}\@ifempty{##1}{}{_{\vphantom{##1}}}%
      \@ifempty{##2}{}{^{##2}}$}%
      \setbox1=\hbox{$####1\vphantom{#2}\@ifempty{##1}{}{_{##1}}%
      \@ifempty{##2}{}{^{\vphantom{##2}}}$}%
      \setbox2=\vbox{\hbox to\wd0{}\hbox to\wd1{}}%
      \mathrel{\hskip\wd2\hskip-\wd0\box0\hskip-\wd1\box1{#2}}%
    }%
    \mathchoice{#3[##1]{##2}}{\next\textstyle}%
    {\next\scriptstyle}{\next\scriptscriptstyle}%
  }%
}

\l@rrow{\xl}{\leftarrow}{\xleftarrow}
\r@rrow{\xr}{\rightarrow}{\xrightarrow}
\l@rrow{\xphl}{\phleftarrow}{\xphleftarrow}
\r@rrow{\xphr}{\phrightarrow}{\xphrightarrow}

\makeatother 

\title{A Characterization of Quasi-Decreasingness\footnote{%
  The research described in this paper is supported by FWF (Austrian
  Science Fund) project P27502.}}
\titlerunning{A Characterization of Quasi-Decreasingness} 

\author[1]{Thomas Sternagel}
\author[1]{Christian Sternagel}
\affil[1]{University of Innsbruck, Austria\\
  \texttt{\{thomas,christian\}.sternagel@uibk.ac.at}}
\authorrunning{T. Sternagel and C. Sternagel} 

\Copyright{Thomas Sternagel and Christian Sternagel}


\makeatletter
\renewcommand*\Copyright[1]{%
  \def\@Copyright{%
      \ifx#1\@empty \else \textcopyright\ #1;\\\fi
    }}
\def\copyrightline{%
  \ifx\@EventLogo\@empty
  \else
    \setbox\@tempboxa\hbox{\includegraphics[height=42\p@]{\@EventLogo}}%
    \rlap{\hspace\textwidth\hspace{-\wd\@tempboxa}\hspace{\z@}%
          \vtop to\z@{\vskip-0mm\unhbox\@tempboxa\vss}}%
  \fi
  \scriptsize
  \vtop{\hsize\textwidth
    \nobreakspace\\
    \@Copyright
    \ifx\@EventLongTitle\@empty\else\@EventLongTitle.\\\fi
    \ifx\@EventEditors\@empty\else
      \@Eds: \@EventEditors
      ; Article~No.\,\@ArticleNo; pp.\,\@ArticleNo:\thepage--\@ArticleNo:\pageref{LastPage}.%
    \fi
  }
}
\makeatother

\EventEditors{Aart Middeldorp and René Thiemann}
\EventNoEds{2}
\EventLongTitle{Proceedings of the 15th International Workshop on
Termination}
\EventShortTitle{WST 2016}
\EventLogo{}
\ArticleNo{5}
\DOIPrefix{}

\theoremstyle{plain}

\newcommand\thmref[1]{Theorem~\ref{thm:#1}}

\newcommand\tabref[1]{Table~\ref{tab:#1}}
\newcommand\lemref[1]{Lemma~\ref{lem:#1}}
\newcommand\defby{\stackrel{\scriptscriptstyle\msf{def}}{=}}
\let\oldeqref\eqref
\renewcommand\eqref[1]{\oldeqref{eq:#1}}
\newcommand\rname{\rho}
\def\systemname#1{\mbox{\textsf{#1}}\xspace}
\newcommand{\muterm}{\systemname{MU-TERM}}
\newcommand{\vmtl}{\systemname{VMTL}}

\newcommand{\TTTT}{%
 \systemname{T\kern-0.2em\raisebox{-0.3em}T\kern-0.2emT\kern-0.2em%
 \raisebox{-0.3em}2}%
}
\newcommand\aprove{\systemname{APro\kern-0.1emVE}}
\newcommand\natt{\systemname{Na\!TT}}

\newcommand\mc[1]{\ensuremath{\mathcal{#1}}\xspace}
\newcommand\msf[1]{\ensuremath{\mathsf{#1}}\xspace}
\newcommand\FF{\mc{F}}
\newcommand\VV{\mc{V}}

\newcommand\RR{\mc{R}}
\newcommand\TT{\mc{T}}
\newcommand\EE{\mc{E}}
\newcommand\PP{\mc{P}}
\newcommand{\Evars}{\ensuremath{\EE\VV}}

\newcommand\fs[1]{\msf{#1}}
\newcommand\sig{\FF}
\newcommand\vars{\VV}
\newcommand\ys{\mathit{ys}}
\newcommand{\Pos}{\ensuremath{\PP\textsf{os}}\xspace}
\newcommand\varsseq{\mathsf{v}}
\newcommand\Evarsseq{\mathsf{ev}}
\newcommand\termsover[2]{\ensuremath{\TT(#1,#2)}}
\newcommand\terms{\termsover{\sig}{\vars}}

\newcommand{\ucs}{{U_\msf{CS}(\RR)}}
\newcommand\IF{\Leftarrow}
\newcommand\EQ{\approx}
\newcommand\crule[1]{\ell_{#1} \to r_{#1} \IF c_{#1}}
\newcommand{\subterm}{\ensuremath{\mathrel{\vartriangleright}}}
\newcommand{\auxrel}{\mathrel{({\succ}\cup{\subterm})^+}}
\newcommand\auxreln[1]{\mathrel{({\succ}\cup{\subterm})^{#1}}}

\begin{document}

\maketitle

\section{Introduction}

In 2010 Schernhammer and Gramlich~\cite{SG10} showed that quasi-decreasingness
of a DCTRS~$\RR$ is equivalent to $\mu$-termination of its context-sensitive
unraveling $\ucs$ on original terms. While the direction that
quasi-decreasingness of $\RR$ implies $\mu$-termination of $\ucs$ on original
terms is shown directly; the converse -- facilitating the use of
context-sensitive termination tools like \muterm~\cite{AGLN10} and
\vmtl~\cite{SG09} -- employs the additional notion of context-sensitive
quasi-reductivity of $\RR$.
In the following, we give a direct proof of the fact that $\mu$-termination of
$\ucs$ on original terms implies quasi-decreasingness of $\RR$.
Moreover, we report our experimental findings on DCTRSs from the confluence
problems database (Cops),\footnote{\url{http://cops.uibk.ac.at}} extending the
experiments of Schernhammer and Gramlich.

\subparagraph{Contribution.}
A direct proof that $\mu$-termination of a CSRS $\ucs$ on original terms implies
quasi-decreasingness of the DCTRS $\RR$. New experiments on a recent DCTRS
collection.

\section{Preliminaries}

We assume familiarity with the basic notions of (conditional and
context-sensitive) term
rewriting~\cite{BN98,L98,O02}, but shortly recapitulate terminology and notation
that we use in the remainder.
Given two arbitrary binary relations $\xr[\alpha]{}$ and $\xr[\beta]{}$,
we write $\xl[\alpha]{}$, $\xr[\alpha]{+}$, $\xr[\alpha]{*}$ for the
\emph{inverse},
the \emph{transitive closure}, and the \emph{reflexive transitive closure}
of $\xr[\alpha]{}$,
respectively.
The relation obtained by considering $\xr[\alpha]{}$ \emph{relative to}
$\xr[\beta]{}$, written $\xr[\alpha/\beta]{}$, is defined by
$\xr[\beta]{*}\cdot\xr[\alpha]{}\cdot\xr[\beta]{*}$.
We use $\vars(\cdot)$ to denote the set of variables occurring in a given
syntactic object, like a term, a pair of terms, a list of terms, etc.
The set of terms $\TT(\FF,\VV)$ over a given signature of function symbols $\FF$
and set of variables $\VV$ is defined inductively:
$x \in \TT(\FF, \VV)$ for all variables $x \in \VV$,
and for every $n$-ary function symbol $f \in \FF$ and terms $t_1,\ldots,t_n \in
\TT(\FF,\VV)$ also $f(t_1,\ldots,t_n) \in \TT(\FF,\VV)$.
A \emph{deterministic oriented 3-CTRS (DCTRS)} $\RR$ is a set of conditional
rewrite rules of the shape $\crule{}$ where $\ell$ and $r$ are terms and $c$ is
a possibly empty sequence of pairs of terms $s_1 \approx t_1, \ldots, s_n
\approx t_n$.
For all rules in $\RR$ we have that
$\ell \not\in \VV$,
$\vars(r) \subseteq \vars(\ell,c)$, and
$\vars(s_i) \subseteq \vars(\ell,t_1,\ldots,t_{i-1})$
for all $1 \leqslant i \leqslant n$.
The rewrite relation induced by a DCTRS~$\RR$ is structured into levels. For
each level $i$, a TRS~$\RR_ i$ is defined recursively by
$\RR_0 = \varnothing$ and
$\RR_{i+1} = \{
  \ell \sigma \approx r \sigma \mid
  \crule{} \in \RR
  \land
  \forall s \approx t \in c.~ s \sigma \xr[\RR_i]{*} t \sigma
\}$
where for a given TRS~$\mathcal{S}$, $\xr[\mathcal{S}]{}$ denotes the induced
rewrite relation (i.e., its closure under contexts and substitutions).
Then the rewrite relation of $\RR$ is $\xr[\RR]{} = \bigcup_{i \geqslant
0}\xr[\RR_i]{}$.
We have $\RR = \RR_\msf{c} \uplus \RR_\msf{u}$ where $\RR_\msf{c}$ denotes the
subset of rules with non-empty conditional part ($n > 0$) and $\RR_\msf{u}$ the
subset of unconditional rules ($n = 0$).
A DCTRS $\RR$ over signature $\FF$ is \emph{quasi-decreasing} if there is a 
well-founded order $\succ$ on $\terms$ such that
${\succ} = ({\succ}\cup{\subterm})^+$,
${\xr[\RR]{}} \subseteq {\succ}$,
and for all rules 
$\ell\to r\Leftarrow s_1\approx t_1,\ldots,s_n\approx t_n$ in $\RR$,
all substitutions $\sigma \colon \VV \to \terms$, and
$0\leqslant i < n$, if
$s_j\sigma \xr[\RR]{*} t_j \sigma$
for all $1 \leqslant j \leqslant i$
then $\ell \sigma \succ s_\text{i+1} \sigma$
.

Given a DCTRS $\RR$ its \emph{unraveling $U(\RR)$} (cf.~\cite[p. 212]{O02}) is
defined as follows.
For each conditional rule 
$\rname\colon \ell\to r\Leftarrow s_1 \approx t_1, \ldots, s_n \approx t_n$
(where $n > 0$)
we introduce $n$ fresh function symbols $U^\rname_1,\ldots,U^\rname_n$
and generate the set of $n + 1$ unconditional rules $U(\rname)$ as follows
\begin{align*}
\ell &\to U^\rname_1(s_1,\varsseq(\ell)) \\
U^\rname_1(t_1,\varsseq(\ell)) &\to
U^\rname_2(s_2,\varsseq(\ell),\Evarsseq(t_1))\\
&\vdots\\
U^\rname_n(t_n,\varsseq(\ell),\Evarsseq(t_1,\ldots,t_{n-1})) &\to r
\end{align*}
where $\varsseq$ and $\Evarsseq$ denote functions that yield the respective
sequences of elements of $\vars$ and $\Evars$ in some arbitrary but fixed order,
and
$\Evars(t_i) = \vars(t_i) \setminus \vars(\ell,t_1,\ldots,t_{i-1})$
denotes the \emph{extra variables} of the right-hand side of the $i$th
condition.
Finally the unraveling of the DCTRS is $U(\RR) = {\RR_\msf{u}}\cup{\bigcup_{\rname \in
\RR_\msf{c}}U(\rname)}$
.

A \emph{context-sensitive rewrite system} (CSRS) is a TRS (over signature $\FF$)
together with a replacement map $\mu \colon \FF \to 2^\mathbb{N}$ that restricts
the argument positions of each function symbol in $\FF$ at which we are allowed
to rewrite.
A position $p$ is \emph{active} in a term $t$ if either $p = \epsilon$,
or $p = iq$, $t = f(t_1,\ldots,t_n)$, $i \in \mu(f)$, and $q$ is active in
$t_i$.
The set of active positions in a term $t$ is denoted by $\Pos_\mu(t)$.
Given a CSRS $\RR$ a term $s$ $\mu$-rewrites to a term $t$, written $s \to_\mu
t$, if $s \to_\RR t$ at some position $p$ and $p \in \Pos_\mu(s)$.
A CSRS is called \emph{$\mu$-terminating} if its context-sensitive rewrite
relation is terminating.
The (proper) subterm relation with respect to replacement map $\mu$, written
$\subterm_\mu$, restricts the ordinary subterm relation to active positions.

We conclude this section by recalling the notion of context-sensitive
quasi-reductivity in an attempt to further appreciation for a proof without this
notion.
\begin{definition}
A CSRS $\RR$ over signature $\FF$ is \emph{context-sensitively quasi-reductive}
if there is an extended signature $\FF' \supseteq \FF$, a replacement map
$\mu$ (with $\mu(f) = \{1, \ldots, n\}$ for every $n$-ary $f \in \FF$), and a
$\mu$-monotonic, well-founded partial order $\succ_\mu$ on $\TT(\FF',\VV)$ such
that for every rule~$\ell \to r \IF s_1 \EQ t_1,\ldots,s_k \EQ t_k$, every
substitution~$\sigma : \VV \to \TT(\FF,\VV)$, and every $0 \leqslant i \leqslant
k-1$:
\begin{itemize}
\item
$\ell\sigma \mathrel{({\succ_\mu}\cup{\subterm_\mu})}^+ s_{i+1}\sigma$ whenever
$s_j\sigma \succeq_\mu t_j\sigma$ for every $1\leqslant j\leqslant i$, and

\item
$\ell\sigma \succ_\mu r\sigma$ whenever $s_j\sigma \succeq_\mu t_j\sigma$
for every $1\leqslant j\leqslant k$.
\end{itemize}
\end{definition}

\section{Characterization}

In order to present our main result (the proof of \thmref{main} below) we first
restate some definitions and theorems which we will use in the proof.

The usual unraveling is extended by a replacement map in order to restrict
reductions in $U$-symbols to the first argument position~\cite[Definition
4]{SG10}.

\begin{definition}[Unraveling $\ucs$]
The \emph{context-sensitive unraveling $\ucs$} is the unraveling $U(\RR)$
together with the replacement map $\mu$ such that 
$\mu(f) = \{1,\ldots,k\}$ if $f\in\FF$ with arity $k$ and 
$\mu(f) = \{1\}$ otherwise.
We say that the resulting CSRS is \emph{$\mu$-terminating on original
terms}~\cite[Definition~7]{SG10}, if there is no infinite $\ucs$-reduction
starting from a term $t \in \terms$.
\end{definition}

Simulation completeness of $\ucs$ (i.e., that every $\RR$-step can be simulated
by a $\ucs$-reduction) can be shown by induction on the level of a
conditional rewrite step~\cite[Theorem~1]{SG10}.
\begin{theorem}[Simulation completeness]
\label{thm:simcomp}
  For a DCTRS $\RR$ we have ${\xr[\RR]{}}\subseteq{\xr[\ucs]{+}}$.
  \qed
\end{theorem}
Furthermore, we need the following auxiliary result.
\begin{lemma}
  \label{lem:one}
  For any context-sensitive rewrite relation $\xr[\mu]{}$ induced by the
  replacement map $\mu$, $\subterm_\mu$ commutes over $\xr[\mu]{}$, i.e.,
  ${\subterm_\mu}\cdot{\xr[\mu]{}} \subseteq {\xr[\mu]{}}\cdot{\subterm_\mu}$.
\end{lemma}
\begin{proof}
  Assume $s \subterm_\mu t \xr[\mu]{} u$ for some terms~$s$, $t$, and $u$. Then
  $s = {C[t]}\subterm_\mu{t} \xr[\mu]{} u$ for some nonempty context~$C$.  Thus
  we conclude by $C[t] \xr[\mu]{} C[u] \subterm_\mu u$.
\end{proof}

%
With this we are finally able to prove our main result.
\begin{theorem}
  \label{thm:main}
  If the CSRS $\ucs$ is $\mu$-terminating on original terms then the DCTRS $\RR$
  is quasi-decreasing.
\end{theorem}
\begin{proof}
Assume that $\ucs$ is $\mu$-terminating on original terms.
We define an order $\succ$ on $\terms$
\[
\textstyle
  {\succ}
  \defby
  {({\xr[\ucs]{}}\cup{\subterm_\mu})^+ \cap (\terms \times \terms)}
  \tag{$\star$}\label{eq:1}
\]
and show that it satisfies the four properties from the definition of
quasi-decreasingness:
\begin{enumerate}
  \item We start by showing that $\succ$ is well-founded on $\terms$.
    Assume, to the contrary, that $\succ$ is not well-founded. Then we have
    an infinite sequence
    \begin{equation}
    {t_1}\succ{t_2}\succ{t_3}\succ{\ldots}\tag{$\dagger$}\label{eq:inf}
    \end{equation}
    where all
    $t_i\in\terms$.
    By definition $\subterm_\mu$ is well-founded.
    Moreover, since $\ucs$ is $\mu$-terminating on original terms, $\xr[\ucs]{}$ is
    well-founded on $\terms$.
    Further note that every $\xr[\ucs]{}$-terminating element
    (hence every term in $\terms$) is
    ${\xr[\ucs/{\subterm_\mu}]{}}$-terminating, since
    by a repeated application of \lemref{one} every infinite reduction
    $
      t_1 \xr[\ucs/{\subterm_\mu}]{} t_2 \xr[\ucs/{\subterm_\mu}]{} \cdots
    $
    starting from a term $t_1 \in \terms$
    can be transformed into an infinite $\xr[\ucs]{}$-reduction, contradicting
    well-foundedness of $\xr[\ucs]{}$ on $\terms$.
    We conclude by analyzing the following two cases:
    \begin{itemize}
      \item
      Either \eqref{inf} contains $\xr[\ucs]{}$ only finitely often,
      contradicting well-foundedness of $\subterm_\mu$,
      \item
      or there are infinitely many $\xr[\ucs]{}$-steps in \eqref{inf}.
      But then we can construct a sequence
      ${s_1}{\xr[\ucs/{\subterm_\mu}]{}}
      {s_2}{\xr[\ucs/{\subterm_\mu}]{}}
      {s_3}{\xr[\ucs/{\subterm_\mu}]{}}
      \ldots$ with $s_1 = t_1$,
      contradicting the fact that all elements of $\terms$ are
      ${\xr[\ucs/{\subterm_\mu}]{}}$-terminating.
    \end{itemize}
  \item Next we show ${\succ}={\auxrel}$. The direction
    ${\succ}\subseteq{\auxrel}$ is obvious.
    For the other direction, ${\auxrel}\subseteq{\succ}$,
    assume we have ${s}\auxreln{n+1}{t}$.
    Then we proceed by induction on $n$.
    In the base case $s \mathrel{({\succ}\cup{\subterm})} t$.
    If $s \succ t$ we are done. Otherwise, $s \subterm t$ and thus also $s
    \subterm_\mu t$ since $s, t\in\terms$ and therefore $s \succ t$.
    In the step case $n = k + 1$ for some $k$, and $s \mathrel{({\succ}\cup{\subterm})} u
    \auxreln{k}{t}$. Then we obtain $s \succ u$ by a similar case-analysis as in
    the base case. Moreover $u \succ t$ by induction hypothesis, and thus $s
    \succ t$.
  \item Now we show that ${\xr[\RR]{}}\subseteq{\succ}$. Assume ${s}\xr[\RR]{}{t}$.
    Together with simulation completeness of $\ucs$, \thmref{simcomp}, we get ${s}\xr[\ucs]{+}{t}$
    which in turn implies ${s}\succ{t}$.
  \item Finally, we show that if for all $\ell \to r \IF s_1 \approx t_1, \ldots,
    s_n \approx t_n$ in $\RR$, substitutions $\sigma\colon \VV \to \terms$, and
    $0 \leqslant i < n$,
    if ${s_j\sigma}\xr[\RR]{*}{t_j\sigma}$ for all $1 \leqslant j \leqslant i$ then
    ${\ell\sigma}\succ{s_{i+1}\sigma}$.
    We have the sequence
    \[\textstyle
      {\ell\sigma}\xr[\ucs]{+}{U^\rname_{i+1}(s_{i+1},\varsseq(\ell),\Evarsseq(t_1,\ldots,t_i))\sigma}
      \subterm_\mu{s_{i+1}\sigma}
    \]
    using the definition of $\ucs$ together with simulation completeness (\thmref{simcomp}).
    But then also ${\ell\sigma}\succ{s_{i+1}\sigma}$ as wanted because
    $\ell\sigma,s_{i+1}\sigma\in\terms$.
\end{enumerate}
Hence $\RR$ is quasi-decreasing with the order $\succ$.
\end{proof}

The converse of \thmref{main} has already been shown by Schernhammer and
Gramlich~\cite[Theorem 4]{SG10}:
\begin{theorem}
  If a DCTRS $\RR$ is quasi-decreasing then the CSRS $\ucs$ is $\mu$-terminating
  on original terms.
  \qed
\end{theorem}
Thus the desired equivalence follows as an easy corollary.
\begin{corollary}
  Quasi-decreasingness of a DCTRS $\RR$ is equivalent to $\mu$-termination of
  the CSRS $\ucs$ on original terms.
\end{corollary}

\section{Experiments}
\label{sec:exps}


\begin{table}[t]
  \centering
  \caption{(Non-)quasi-decreasing DCTRSs out of 103 in Cops
  by transformation and tool.}
  \begin{tabularx}{\textwidth}{Xc@{\ }c@{\ }cc@{\ }c@{\ }cc@{\ }c@{\ }c@{\ }c@{\ }cc}
    \toprule
    & \multicolumn{3}{c}{conditional $\RR$}
    & \multicolumn{3}{c}{$\ucs$}
    & \multicolumn{5}{c}{$U(\RR)$} & \\
    \cmidrule(lr){2-4}
    \cmidrule(lr){5-7}
    \cmidrule(lr){8-12}
    & \scriptsize \aprove & \scriptsize \muterm & \scriptsize \vmtl
    & \scriptsize \aprove & \scriptsize \muterm & \scriptsize \vmtl
    & \scriptsize \aprove & \scriptsize \muterm & \scriptsize \natt
    & \scriptsize \TTTT & \scriptsize \vmtl & total \\
    \midrule
    YES & 80 & 78 & 80 & 78 & 78 & 79 & 81 & 78 & 77 & 78 & 78 & 84 \\
    NO & -- & 12 & -- & -- & -- & -- & -- & -- & -- & -- & -- & 12 \\
    \bottomrule
  \end{tabularx}
  \label{tab:exp}
\end{table}

In order to present up-to-date numbers for (non-)quasi-decreasingness we
conducted experiments on the 103 DCTRSs contained in the confluence problems
database using various automated termination tools.
Of these, \aprove~\cite{GBE14},
\muterm~5.13~\cite{AGLN10}, and \vmtl~1.3~\cite{SG09} are able to directly show
quasi-decreasingness and
\muterm is the only tool that can show non-quasi-decreasingness~\cite{LMG14}.
\aprove, \muterm, and \vmtl can also handle context-sensitive
systems and we used them in combination with $\ucs$.
Finally, we also ran the previous tools together with
\natt~\cite{YKS14} and \TTTT~1.16~\cite{KSZM09} on
$U(\RR)$.
The results for a timeout of one minute are shown in \tabref{exp}.
There are several points of notice.
%
The most yes-instances (81) we get if we use \aprove together with $U(\RR)$.
Interestingly, \aprove cannot show quasi-decreasingness of system 362 directly,
although it succeeds (like all other tools besides \natt) if provided with its
unraveling.
Moreover, systems 266, 278, and 279 can be shown to be quasi-decreasing by \aprove if
we use $U(\RR)$ but not if we use $\ucs$ (even if we increase the timeout to 5
minutes).
On system 363 only \muterm succeeds (in the direct approach).
If we compare \muterm on conditional systems to \muterm with $\ucs$,
the direct method succeeds on system 360 but not
on system 329. Conversely, when using $\ucs$ it succeeds on system 329
but not on system 360.
Moreover, \muterm seems to have some problems with systems 278 and 342,
generating errors in the direct approach.
With $\ucs$ \vmtl succeeds on 79 systems, subsuming the results from \aprove and
\muterm (78 each). On system 357 only \vmtl together with $\ucs$ succeeds.
With $U(\RR)$, \natt succeeds on 77 systems, this is subsumed by
\TTTT, succeeding on 78 systems, which in turn is subsumed by \aprove,
succeeding, as mentioned above, on 81 systems.
In total 84 systems are shown to be quasi-decreasing, 12 systems to be
non-quasi-decreasing, and only 7 remain open.
%
%
One of these, for example, is system 337 from Cops, for computing
Bubble-sort~\cite{SR06}
\begin{xalignat*}{2}
  x < 0 &\to \fs{false} &
  0 < \fs{s}(y) &\to \fs{true}\\
  \fs{s}(x) < \fs{s}(y) &\to x < y &
  x : y : \ys &\to y : x : \ys \IF x < y \EQ \fs{true}
\intertext{whose unraveling replaces the last (and only conditional) rule by the
two rules:}
  x : y : \ys &\to \fs{U}(x < y, x, y, \ys) &
  \fs{U}(\fs{true}, x, y, \ys) &\to y : x : \ys
\end{xalignat*}

\section{Conclusion}

We provide a direct proof
for one direction of a previous characterization of quasi-decreasingness, i.e.,
that $\mu$-termination of a CSRS
$\ucs$ on original terms implies quasi-decreasingness of the DCTRS $\RR$ without
the need of a detour by using the notion of context-sensitive quasi-reductivity.
We believe that our proof could easily be adapted to any other context-sensitive
transformation as long as it is simulation complete.
Moreover, we provide experimental results on a recent collection of DCTRSs.
Knowing that a DCTRS is quasi-decreasing is, among other things, useful to
show confluence with the Knuth-Bendix criterion for CTRSs~\cite{AL94}.

\subparagraph*{Acknowledgments.}

We thank the Austrian Science Fund (FWF project P27502) for supporting
our work.
Moreover we would like to thank the anonymous reviewers for useful hints and
remarks and particularly for pointing out a flaw in an earlier version of
Section~\ref{sec:exps}.



\bibliography{short,references}


\end{document}